\documentclass[letterpaper]{article} 
\usepackage{times}  
\usepackage{helvet}  
\usepackage{courier}  
\usepackage[hyphens]{url}  
\usepackage{graphicx} 
\usepackage{natbib}  

\usepackage{caption} 
\DeclareCaptionStyle{ruled}{labelfont=normalfont,labelsep=colon,strut=off} 
\usepackage{algorithm}
\usepackage[pdftex, colorlinks,
  linkcolor=blue,
  citecolor=black,
  filecolor=blue,urlcolor=blue]%
{hyperref}
\usepackage{newfloat}
\usepackage{listings}
\floatstyle{ruled}
\newfloat{listing}{tb}{lst}{}
\floatname{listing}{Listing}

\usepackage{amsfonts}
\usepackage{nicefrac}       
\usepackage{amsthm,bm}
\usepackage{mathtools}
\usepackage{amsmath}
\usepackage{mathrsfs}
\usepackage{amssymb}
\usepackage{xcolor}
\usepackage{cleveref}
\usepackage{subfigure}
\usepackage{wrapfig}
\usepackage{lipsum}
\usepackage{paralist}
\usepackage[noend]{algpseudocode}
\newtheorem{theorem}{Theorem}[section]
\newtheorem{lemma}[theorem]{Lemma}
\newtheorem{definition}[theorem]{Definition}

\newtheorem{claim}[theorem]{Claim}

\theoremstyle{definition}
\newtheorem{example}[theorem]{Example}
\newtheorem{remark}[theorem]{Remark}

\newcommand{\bmx}{{\bm x}}
\newcommand{\bmy}{{\bm y}}

\newcommand{\bmw}{{\bm w}}

\newcommand{\bmone}{{\bm 1}}

\newcommand{\bbR}{\mathbb{R}}

\newcommand{\E}{\mathbb{E}}
\newcommand{\bbP}{\mathbb{P}}

\newcommand{\mc}[1]{\mathcal{#1}}

\newcommand{\ip}[1]{\langle{ #1 \rangle} }

\newcommand{\argmax}{\mathop{\mathrm{argmax}}}

\title{Sparsification of Decomposable Submodular Functions\footnote{An extended abstract of this work will appear in the \emph{Proceedings of the Thirty-Sixth AAAI Conference on Artificial Intelligence (AAAI 2022)}.}}
\author {
    Akbar Rafiey\thanks{Department of Computing Science, Simon Fraser University, Burnaby, BC, Canada. {Email: \texttt{arafiey@sfu.ca}}} \and
    Yuichi Yoshida \thanks{National Institute of Informatics, Tokyo, Japan. {Email: \texttt{yyoshida@nii.ac.jp}}}
}

\date{}
\begin{document}

\maketitle

\begin{abstract}
    Submodular functions are at the core of many machine learning and data mining tasks.
    The underlying submodular functions for many of these tasks are decomposable, i.e., they are sum of several simple submodular functions. In many data intensive applications, however, the number of underlying submodular functions in the original function is so large that we need prohibitively large amount of time to process it and/or it does not even fit in the main memory.
    To overcome this issue, we introduce the notion of sparsification for decomposable submodular functions whose objective is to obtain an accurate approximation of the original function that is a (weighted) sum of only a few submodular functions.
    Our main result is a polynomial-time randomized sparsification algorithm such that the expected number of functions used in the output is independent of the number of underlying submodular functions in the original function.
    We also study the effectiveness of our algorithm under various constraints such as matroid and cardinality constraints.
    We complement our theoretical analysis with an empirical study of the performance of our algorithm.
\end{abstract}

\newpage

\section{Introduction}
Submodularity of a set function is an intuitive diminishing returns property, stating that adding an element to a smaller set helps gaining more return than adding it to a larger set. This fundamental structure has emerged as a very beneficial property in many combinatorial optimization problems arising in machine learning, graph theory, economics, game theory, to name a few. 
Formally, a set function $f\colon 2^E \to \mathbb{R}$ is \emph{submodular} if for any $S \subseteq T \subseteq E$ and $e \in E\setminus T$ it holds that
\[
    f(S \cup \{e \}) - f(S) \geq f(T \cup \{e\}) - f(T).
\]

Submodularity allows one to efficiently find provably (near-)optimal solutions. In particular, a wide range of efficient approximation algorithms have been developed for maximizing or minimizing submodular functions subject to different constraints. Unfortunately, these algorithms require number of function evaluations which in many data intensive applications are infeasible or extremely inefficient. Fortunately, several submodular optimization problems arising in machine learning have structure that allows solving them more efficiently. A novel class of submodular functions are \textit{decomposable} submodular functions. These are functions that can be written as sums of several ``simple'' submodular functions, i.e.,
\[
    F(S) = \sum_{i=1}^N f_i(S) \quad \quad \forall S\subseteq E,
\]
where each $f_i\colon 2^E \to \mathbb{R}$ is a submodular function on the ground set $E$ with $|E|=n$.

Decomposable submodular functions encompass many of the examples of submodular functions
studied in the context of machine learning as well as economics. For example, they are extensively used in economics in the problem of welfare maximization in combinatorial auctions~\cite{DobzinskiS06,Feige06,FeigeV06,PapadimitriouSS08,Vondrak08}.
\begin{example}[Welfare maximization]
Let $E$ be a set of $n$ resources and $a_1,\dots,a_N$ be $N$ agents. Each agent has an interest over subsets of resources which is expressed as a submodular function $f_i:2^E\to \bbR$. The objective is to select a small subset of resources that maximizes the happiness across all the agents, the ``social welfare''. More formally, the goal is to find a subset $S\subseteq E$ of size at most $k$ that maximizes $F(S)=\sum_{i=1}^N f_i(S)$, where $k$ is a positive integer.
\end{example}
Decomposable submodular functions appear in various machine learning tasks such as data summarization, where we seek a representative subset of elements of small size. This has numerous applications, including exemplar-based clustering~\cite{DueckF07,GomesK10}, image summarization~\cite{TschiatschekIWB14}, recommender systems~\cite{ParambathUG16}, and document and corpus summarization \cite{lin2011class}. The problem of maximizing decomposable submodular functions has been studied under different constraints such as cardinality and matroid constraints in various data summarization settings \cite{MirzasoleimanBK16,MirzasoleimanKS16,MirzasoleimanZK16}, and differential privacy settings \cite{ChaturvediNZ21,MitrovicB0K17,RafieyY20}.


In many of these applications, the number of underlying submodular functions are too large (i.e., $N$ is too large) to even fit in the main memory, and building a compressed representation that preserves relevant properties of the submodular function is appealing. This motivates us to find a \emph{sparse} representation for a decomposable submodular function $F$. Sparsification is an algorithmic paradigm where a dense object is replaced by a sparse one with similar ``features'', which often leads to significant improvements in efficiency of algorithms, including running time, space complexity, and communication. In this work, we propose a simple and very effective algorithm that yields a sparse and accurate representation of a decomposable submodular function. To the best of our knowledge this work is the first to study sparsification of decomposable submodular functions.


\subsection{Our contributions}

\paragraph{General setting.} Given a decomposable submodular function $F=\sum_{i=1}^N f_i$, we present a randomized algorithm that yields a sparse representation that approximates $F$. Our algorithm chooses each submodular function $f_i$ with probability proportional to its ``importance'' in the sum $\sum_{i=1}^N f_i$ to be in the sparsifier. Moreover, each selected submodular function will be assigned a weight which also is proportional to its ``importance''. We prove this simple algorithm yields a sparsifier of small size (\emph{\textbf{independent of $N$}}) with a very good approximation of $F$. Let $|\mc{B}(f_i)|$ denote the number of extreme points in the base polytope of $f_i$, and $B=\max_{i\in [N]} |\mc{B}(f_i)|$.
\begin{theorem}\label{thm:main-body-S}
    Let $F=\sum_{i=1}^N f_i$ be a decomposable submodular function. For any $\epsilon >0$, there exists a vector $\bmw\in \bbR^N$ with at most $O(\frac{B \cdot n^2}{\epsilon^2})$ non-zero entries such that for the submodular function $F'=\sum_{i=1}^N \bmw_i f_i$ we have
    \[
        (1- \epsilon)F'(S)\leq F(S)\leq (1+ \epsilon)F'(S) \quad \forall S\subseteq E.
    \]
     Moreover, if all $f_i$'s are monotone, then there exists a polynomial-time randomized algorithm that outputs a vector $\bmw\in \bbR^N$ with at most $O(\frac{B\cdot n^{2.5} \log{n}}{\epsilon^2})$ non-zero entries in expectation such that for the submodular function $F'=\sum_{i=1}^N \bmw_i f_i$, with high probability, we have
    \[
        (1- \epsilon)F'(S)\leq F(S)\leq (1+ \epsilon)F'(S) \quad \forall S\subseteq E.
    \]
\end{theorem}
\begin{remark}[Tightness]
The existential result is almost tight because in the special case of directed graphs, we have $\max_i |\mc{B}(f_i)| = 2$ and it is known that we need $\Omega(n^2)$ edges to construct a sparsifier~\cite{Cohen2017}. 
\end{remark}

\paragraph{Sparsifying under constraints.} We consider the setting where we only are interested in evaluation of $F$ on particular sets. For instance, the objective is to optimize $F$ on subsets of size at most $k$, or it is to optimize $F$ over subsets that form a matroid. Optimizing submodular functions under these constraints has been extensively studied and has an extremely rich theoretical landscape. Our algorithm can be tailored to these types of constraints and constructs a sparsifier of even smaller size.

\begin{theorem}\label{thm:main-body-K-S}
    Let $F=\sum_{i=1}^N f_i$ be a decomposable submodular function. For any $\epsilon >0$ and a matroid $\mc{M}$ of rank $r$, there exists a vector $\bmw\in \bbR^N$ with at most $O(\frac{B\cdot r\cdot n}{\epsilon^2})$ non-zero entries such that for the submodular function $F'=\sum_{i=1}^N \bmw_i f_i$ we have
    \[
        (1- \epsilon)F'(S)\leq F(S)\leq (1+ \epsilon)F'(S) \quad \forall S\subseteq \mc{M}.
    \]
    Moreover, if all $f_i$'s are monotone, then there exists a polynomial-time randomized algorithm that outputs a vector $\bmw\in \bbR^N$ with at most $O(\frac{B\cdot r\cdot n^{1.5} \log n}{\epsilon^2})$ non-zero entries in expectation such that for the submodular function $F'=\sum_{i=1}^N \bmw_i f_i$, with high probability, we have
    \[
        (1- \epsilon)F'(S)\leq F(S)\leq (1+ \epsilon)F'(S) \quad \forall S\subseteq \mc{M}.
    \]
\end{theorem}

\paragraph{Applications, speeding up maximization/minimization.} 
Our sparsifying algorithm can be used as a preprocessing step in many settings in order to speed up algorithms. To elaborate on this, we consider the classical greedy algorithm of~\citeauthor{NemhauserWF78} for maximizing monotone submodular functions under cardinality constraints. We observe that sparsifying the instance reduces the number of function evaluations from $O(k n N)$ to $O(\frac{B k^2 n^{2}}{\epsilon^2})$, which is a significant speed up when $N\gg n$. Regarding minimization, we prove our algorithm gives an approximation on the \emph{Lov\'{a}sz extension}, thus it can be used as a preprocessing step for algorithms working on Lov\'{a}sz extensions such as the ones in \cite{AxiotisK0SV21,EneNV17}. One particular regime that has been considered in many results is where each submodular function $f_i$ acts on $O(1)$ elements of the ground set which implies $B=\max_i|\mc{B}(f_i)|$ is $O(1)$. 
Using our sparsifier algorithm as a preprocessing step is quite beneficial here. 
For instance, it improves the running time of \citet{AxiotisK0SV21} from $\widetilde{O}(T_{\rm maxflow}(n,n+N)\log{\frac{1}{\epsilon}})$ to $\widetilde{O}(T_{\rm maxflow}(n,n+\frac{n^2}{\epsilon^2})\log{\frac{1}{\epsilon}})$. Here, $T_{\rm maxflow}(n,m)$ denotes the time required to compute the maximum flow in a directed graph of $n$ vertices and $m$ arcs with polynomially bounded integral capacities.

\paragraph{Well-known examples.} In practice, the bounds on the size of sparsifiers are often better than the ones presented in \Cref{thm:main-body-S,thm:main-body-K-S} e.g. $B$ is a constant. We consider several examples of decomposable submodular functions that appear in many applications, namely, \textsc{Maximum Coverage}, \textsc{Facility Location}, and \textsc{Submodular Hypergraph Min Cut} problems. For the first two examples, sparsifiers of size $O(\frac{n^2}{\epsilon^2})$ can be constructed in time linear in $N$. For \textsc{Submodular Hypergraph Min Cut} when each hyperedge is of constant size sparsifiers of size $O(\frac{n^2}{\epsilon^2})$ exist, and in several specific cases with various applications efficient algorithms are employed to construct them. 
\paragraph{Empirical results.} Finally, we empirically examine our algorithm and demonstrate that it constructs a concise sparsifier on which we can efficiently perform algorithms.


\subsection{Related work}
To the best of our knowledge there is no prior work on sparsification algorithms for decomposable submodular functions. However, special cases of this problem have attracted much attention, most notably \emph{cut sparsifiers} for graphs. The cut function of a graph $G=(V,E)$ can be seen as a decomposable submodular function $F(S) = \sum_{e \in E}f_e$, where $f_e(S) = 1$ if and only if $e \cap S \neq \emptyset$  and $e \cap (V \setminus S) \neq \emptyset$.
The problem of sparsifying a graph while approximately preserving its cut structure has been extensively studied,
(See~\citet{AhnGM12,AhnGM13,BansalST19,BenczurK15} and references therein.)
The pioneering work of~\citet{BenczurK96} showed for any graph $G$ with $n$ vertices one can construct a weighted subgraph $G'$ in nearly linear time with $O(n\log n/\epsilon^2)$ edges such that the weight of every cut in $G$ is preserved  within a multiplicative $(1 \pm \epsilon)$-factor in $G'$. Note that a graph on $n$ vertices can have $N=\Omega(n^2)$ edges. The bound on the number of edges was later improved to $O(n/\epsilon^2)$~\cite{BatsonSS12} which is tight~\cite{AndoniCKQWZ16}. 

A more general concept for graphs called \emph{spectral sparsifier} was introduced by~\citet{SpielmanT11}. This notion captures the spectral similarity between a graph and its sparsifiers. A spectral sparsifier approximates the \textit{quadratic form of the Laplacian} of a graph. Note that a spectral sparsifier is also a cut sparsifier. This notion has numerous applications in linear algebra~\cite{Mahoney11,LiMP13,CohenLMMPS15,LeeS14}, and it has been used to design efficient approximation algorithms related to cuts and flows~\cite{BenczurK15,KargerL02,Madry10}. Spielman and Teng's sparsifier has $O(n\log^c n)$ edges for a large constant $c>0$ which was improved to $O(n/\epsilon^2)$~\cite{Lee018}.

In pursuing a more general setting, the notions of cut sparsifier and spectral sparsifier have been studied for \emph{hypergraphs}. Observe that a hypergraph on $n$ vertices can have exponentially many hyperedges i.e., $N = \Omega(2^n)$. For hypergraphs, \citet{KoganK15} provided a polynomial-time algorithm that constructs an $\epsilon$-cut sparsifier with $O(n(r + \log n)/\epsilon^2)$ hyperedges where $r$ denotes the maximum size of a hyperedge. The current best result is due to \citet{ChenKN20} where their $\epsilon$-cut sparsifier uses $O(n\log n/\epsilon^2)$ hyperedges and can be constructed in time $O(Nn^2+n^{10}/\epsilon^2)$ where $N$ is the number of hyperedges.
Recently, \citet{Soma2019} initiated the study of spectral sparsifiers for hypergraphs and showed that every hypergraph admits an $\epsilon$-spectral sparsifier with $O(n^3\log n/\epsilon^2)$ hyperedges. For the case where the maximum size of a hyperedge is $r$, \citet{BansalST19} showed that every hypergraph has an $\epsilon$-spectral sparsifier of size $O(nr^3\log n/\epsilon^2)$. Recently, this bound has been improved to $O(nr(\log n/\epsilon)^{O(1)})$ and then to $O(n(\log n/\epsilon)^{O(1)})$ ~\cite{KapralovKTY2021b,KapralovKTY2021}.
This leads to the study of sparsification of submodular functions which is the focus of this paper and provides a unifying framework for these previous works.

\section{Preliminaries}\label{sec:pre}
For a positive integer $n$, let $[n] = \{1,2,\ldots,n\}$.
Let $E$ be a set of elements of size $n$ which we call the \emph{ground set}. For a set $S \subseteq E$, $\bmone_S \in \mathbb{R}^E$ denotes the characteristic vector of $S$.
For a vector $\bmx \in \mathbb{R}^E$ and a set $S \subseteq E$, $\bmx(S)=\sum_{e \in S}\bmx(e)$.

\paragraph{Submodular functions.}
Let $f\colon 2^E \to \mathbb{R}_+$ be a set function.
We say that $f$ is \emph{monotone} if $f(S) \leq f(T)$ holds for every $S \subseteq T \subseteq E$.
We say that $f$ is \emph{submodular} if $f(S \cup \{e \}) - f(S) \geq f(T \cup \{e\}) - f(T)$ holds for any $S \subseteq T \subseteq E$ and $e \in E\setminus T$. The \emph{base polytope} of a submodular function $f$ is defined as
\[
    \mc{B}(f)=\{ \bmy\in \mathbb{R}^E \mid \bmy(S)\leq f(S) ~~ \forall S\subseteq E, \bmy(E)=f(E)\},
\]
and $|\mc{B}(f)|$ denotes the number of \emph{extreme points} in the base polytope $\mc{B}(f)$.

\begin{definition}[$\epsilon$-sparsifier]
    Let $f_i$ ($i\in D$) be a set of $N$ submodular functions, and $F(S)=\sum_{i\in D} f_i(S)$ be a decomposable submodular function.
    A vector $\bmw \in \bbR^N$ is called an $\epsilon$-sparsifier of $F$ if, for the submodular function $F' := \sum_{i\in D} \bmw_i f_i$, the following holds for every $S\subseteq E$
        \begin{align}
        \label{eq:S}
            (1-\epsilon)F'(S) \leq F(S) \leq (1+\epsilon)F'(S).
        \end{align}
    The \emph{size} of an $\epsilon$-sparsifier $\bmw$, $\mathrm{size}(\bmw)$, is the number of indices $i$'s with $\bmw_i\neq 0$.
\end{definition}

\paragraph{Matroids and matroid polytopes.}

A pair $\mathcal{M} = (E,\mc{I})$ of a set $E$ and $\mc{I}\subseteq 2^E$ is called a \emph{matroid} if
\begin{inparaenum}
\item[(1)] $\emptyset \in \mc{I}$,
\item[(2)] $A\in \mc{I}$ for any $A \subseteq B\in \mc{I}$, and
\item[(3)] for any $A,B\in \mc{I}$ with $|A| < |B|$, there exists $e \in B\setminus A$ such that $A\cup \{e\}\in \mc{I}$.
\end{inparaenum}
We call a set in $\mathcal{I}$ an \emph{independent set}. The \emph{rank function} $r_{\mc{M}}\colon 2^E \to \mathbb{Z}_+$ of $\mc{M}$ is $r_{\mc{M}}(S)=\max\{|I|:I\subseteq S, I\in \mc{I}\}$.
An independent set $S\in \mc{I}$ is called a \emph{base} if $r_{\mc{M}}(S)=r_{\mc{M}}(E)$. We denote the rank of $\mc{M}$ by $r(\mc{M})$.
The \emph{matroid polytope} $\mc{P}(\mc{M}) \subseteq \bbR^E$ of $\mc{M}$ is
$
  \mc{P}(\mc{M})=\mathrm{conv}\{\mathbf{1}_I : I\in \mc{I}\},
$
where $\mathrm{conv}$ denotes the convex hull.
Or equivalently~\cite{edmonds2003submodular},
$
  \mc{P}(\mc{M})=\left\{\bmx\geq \mathbf{0} : \bmx(S)\leq r_{\mc{M}}(S) \; \forall S\subseteq E \right\}.
$

\paragraph{Concentration bound.}
We use the following concentration bound:
\begin{theorem}[Chernoff bound, see e.g.~\citet{bookMotwaniR95}]\label{lem:Chernoff}
    Let $X_1,\dots,X_n$ be independent random variable in range $[0,a]$. Let $T=\sum_{i=1}^n X_i$. Then for any $\epsilon\in [0,1]$ and $\mu\geq \E[T]$,
    \[
        \bbP\left[ |T-\E[T]| \geq \epsilon\mu\right] \leq 2\exp{\left(-\nicefrac{\epsilon^2\mu}{3a}\right)}.
    \]
\end{theorem}


\section{Constructing a sparsifier}
In this section, we propose a probabilistic argument that proves the existence of an accurate sparsifier and turn this argument into an (polynomial-time) algorithm that finds a sparsifier with high probability.

For each submodular function $f_i$, let
    \begin{align}
    \label{eq:p_i}
        p_i = \max\limits_{A\subseteq E}\frac{f_i(A)}{F(A)}. 
    \end{align}
The values $p_i$'s are our guide on how much weight should be allocated to a submodular function $f_i$ and with what probability it might happen. To construct an $\epsilon$-sparsifier of $F$, for each submodular function $f_i$, we assign weight $1/(\kappa\cdot p_i)$ to $\bmw_i$ with probability $\kappa\cdot p_i$ and do nothing for the complement probability $1-\kappa\cdot p_i$ (see \Cref{alg:sparsification}). Here $\kappa$ depends on $n$, $\epsilon$ and $\delta$ where $\delta$ is the failure probability of our algorithm. Observe that, for each $f_i$, the expected weight $\bmw_i$ is exactly one. We show that the expected number of entries of $\bmw$ with $\bmw_i>0$ is $n^2\cdot\max_{i\in D}|\mc{B}(f_i)|$. Let $B=\max_{i\in D}|\mc{B}(f_i)|$ in the rest of the paper.

\begin{algorithm}[t]
  \caption{}\label{alg:sparsification}
  \begin{algorithmic}[1]
    \Require{Submodular functions $f_i$ in dataset $D$ where each $f_i:\{0,1\}^E\to \bbR$, $\epsilon,\delta \in (0,1)$}
    \State{$\bmw\gets \mathbf{0}$}
    \State{$\kappa\gets 3\log (2^{n+1}/\delta)/\epsilon^2$}
    \For{$f_i$ in $D$}
        \State{$p_i\gets \max_{A\subseteq E} f_i(A)/F(A)$}
        \State{$\kappa_i\gets \min\{1, \kappa\cdot p_i\}$}
        \State{$\bmw_i \gets 1/\kappa_i$ with probability $\kappa_i$}\Comment{do nothing with probability $1-\kappa_i$}
    \EndFor
    \State{\Return $\bmw\in \bbR^D$.}
  \end{algorithmic}
\end{algorithm}

\begin{lemma}\label{lem:SS}
\Cref{alg:sparsification} returns $\bmw$ which is an $\epsilon$-sparsifier of $F$ with probability at least $1-\delta$. 
\end{lemma}

\begin{proof}
    We prove that for every $S\subseteq E$ with high probability it holds that $(1-\epsilon)F'(S) \leq F(S) \leq (1+\epsilon)F'(S)$. 
        
    Observe that by our choice of $p_i$ and $\bmw_i$ we have $\E[F'(S)]=F(S)$, for all subsets $S\subseteq E$. Consider a subset $S_k$. Using \Cref{lem:Chernoff}, we have
    \begin{align}
    \nonumber
        &\bbP\Big[ |F'(S_k)-\E[F'(S_k)]| \geq \epsilon\E[F'(S_k)]\Big] \\
        &= \bbP\Big[ |F'(S_k)-F(S_k)| \geq \epsilon F(S_k)\Big] \\
        &\leq 2\exp{\left(\nicefrac{-\epsilon^2 F(S_k)}{3a}\right)}
        \label{eq:concentration}
    \end{align}
    where $a=\max_{i}\bmw_i f_i(S_k)$. We bound the right hand side of~\eqref{eq:concentration} by providing an upper bound for $a$.
    \begin{align}
    \nonumber
    \label{eq:upper-bound-a}
        a &=\max\limits_{i}\bmw_i f_i(S_k) = \max\limits_{i} \frac{f_i(S_k)}{\kappa\cdot p_i}\\
        &= \max\limits_{i} \frac{f_i(S_k)}{\kappa\cdot\max\limits_{A\subseteq E} \frac{f_i(A)}{F(A)}}\\
        &\leq \max\limits_{i} \frac{f_i(S_k)}{ \kappa\cdot\frac{f_i(S_k)}{F(S_k)}}
        = \frac{F(S_k)}{\kappa}
    \end{align}
    Given the above upper bound for $a$ and the inequality in~\eqref{eq:concentration} yields
    \begin{align*}
         &\bbP\Big[ |F'(S_k)-F(S_k)| \geq \epsilon F(S_k)\Big] \leq 2\exp{\left(-\frac{\epsilon^2 F(S_k)}{3a}\right)}\\ 
         &\leq 2\exp{\left(-\frac{\epsilon^2 F(S_k)}{3 F(S_k)/\kappa}\right)} = 2\exp{\left(\nicefrac{-\kappa\epsilon^2}{3}\right)}.
    \end{align*}
    Recall that $\kappa= 3\log (2^{n+1}/\delta)/\epsilon^2$. 
    Hence, taking a union bound  over all $2^n$ possible subsets yields that Algorithm~\ref{alg:sparsification} with probability at least $1-\delta$ returns a spectral sparsifier for $F$.
\end{proof}

\begin{lemma}\label{lem:UpperBound-p_i}
    \Cref{alg:sparsification} outputs an $\epsilon$-sparsifier with the expected size $O(\frac{B\cdot n^2}{\epsilon^2})$.
\end{lemma}
\begin{proof}
    In \Cref{alg:sparsification}, each $\bmw_i$ is greater than zero with probability $\kappa_i$ and it is zero with probability $1-\kappa_i$. Hence,
    \begin{align}
    \label{eq:up-expected-size}
        \E[\mathrm{size}(\bmw)] &= \sum_{i\in D} \kappa_i \leq  \kappa\sum_{i\in D} p_i
        \leq O\left(\nicefrac{n}{\epsilon^2}\right)\sum_{i\in D} p_i
    \end{align}
    It suffices to show an upper bound for $\sum_{i\in D} p_i$. 
    \begin{claim}
    \label{claim:up-p_i}
        $\sum_{i\in D} p_i\leq n\cdot \max_{i\in D}|\mc{B}(f_i)|= n\cdot B$.
    \end{claim}
    Claim~\ref{claim:up-p_i} and inequality \eqref{eq:up-expected-size} yield the desired bound.
\end{proof}
\Cref{lem:SS,lem:UpperBound-p_i} proves the existence part of \Cref{thm:main-body-S}. That is, for every $\epsilon,\delta\in(0,1)$, there exists an $\epsilon$-sparsifier of size at most $O(\frac{B\cdot n^2}{\epsilon^2})$ with probability at least $1-\delta$.

\paragraph{Polynomial time algorithm.} Observe that computing $p_i$'s~\eqref{eq:p_i} may not be a polynomial-time task in general. However, to guarantee that \Cref{alg:sparsification} outputs an $\epsilon$-sparsifier with high probability it is sufficient to instantiate it with an upper bound for each $p_i$ (see proof of Lemma~\ref{lem:SS}). Fortunately, the result of~\citet{BaiIWB16} provides an algorithm to approximate the ratio of two monotone submodular functions.
\begin{theorem}[\citet{BaiIWB16}]\label{thm:unconstrained-pi-approximation}
    Let $f$ and $g$ be two monotone submodular functions. Then there exists a polynomial-time algorithm that approximates $\max_{S\subseteq E}\frac{f(S)}{g(S)}$ within $O(\sqrt{n}\log n)$ factor. 
\end{theorem}
Hence, when all $f_i$'s are monotone we can compute $\hat{p}_i$'s with $p_i\leq \hat{p}_i\leq O(\sqrt{n}\log n)p_i$ in polynomial time which leads to a polynomial-time randomized algorithm that constructs an $\epsilon$-sparsifier of the expected size at most $O(\frac{B\cdot n^{2.5} \log n}{\epsilon^2})$. This proves the second part of \Cref{thm:main-body-S}.

As we will see, in various applications, the expected size of the sparsifier is often much better than the ones presented in this section. Also, we emphasize that once a sparsifier is constructed it can be reused many times (possibly for maximization/minimization under several different constraints). Hence computing or approximating $p_i$'s should be regarded as a preprocessing step, see Example \ref{ex:knapsack} for a motivating example. Finally, it is straightforward to adapt our algorithm to sparsify decomposable submodular functions of the form $\sum_{i\in D} \alpha_i f_i$, known as \emph{mixtures of submodular functions} \cite{BairiIRB15,TschiatschekIWB14}. 

\begin{example}[Knapsack constraint]
\label{ex:knapsack}
Consider the following optimization problem
     \begin{align}
         \max_{S\subseteq I}\{F(S): \sum_{i\in S}c_i \leq B\}
     \end{align}
     where $I=\{1,\dots,n\}$, $B$ and $c_i$, $i\in I$, are nonnegative integers. In scenarios where the items costs $c_i$ or $B$ are dynamically changing and $F=\sum_{i=1}^{N} f_i$ is decomposable, it is quite advantageous to use our sparsification algorithm and reuse a sparsifier. That is, instead of maximizing $F$ whenever item costs or $B$ are changed, we can maximize $F'$, a sparsification of $F$.
\end{example}


\section{Constructing a sparsifier under constraints}

Here we are interested in constructing a sparsifier for a decomposable submodular function $F$ while the goal is to optimize $F$ subject to constraints. One of the most commonly used and general constraints are matroid constraints. That is, for a matroid $\mathcal{M}=(E,\mathcal{I})$, the objective is finding $S^*=\argmax_{S\subseteq E,S\in \mathcal{I}} F(S)$.
\begin{algorithm}[t]
  \caption{}\label{alg:k-sparsification}
  \begin{algorithmic}[1]
    \Require{Submodular functions $f_i:\{0,1\}^E\to \bbR$ in dataset $D$, matroid $\mathcal{M}=(E,\mathcal{I})$, and $\epsilon,\delta \in (0,1)$}
    \State{$\bmw\gets \mathbf{0}$}
    \State{$\kappa\gets 3\log (2n^{r+1}/\delta)/\epsilon^2$, where $r$ is the rank of $\mathcal{M}$.}
    \For{$f_i$ in $D$}
        \State{$p_i\gets \max_{A\in\mathcal{I}} f_i(A)/F(A)$}
        \State{$\kappa_i\gets \min\{1, \kappa\cdot p_i\}$}
        \State{$\bmw_i \gets 1/\kappa_i$ with probability $\kappa_i$}\Comment{do nothing with probability $1-\kappa_i$}
    \EndFor
    \State{\Return $\bmw\in \bbR^D$.}
  \end{algorithmic}
\end{algorithm}

In this setting it is sufficient to construct a sparsifier that approximates $F$ only on independent sets. It turns out that we can construct a \emph{smaller} sparsifier than the one constructed to approximate $F$ everywhere. For each submodular function $f_i$, let
    \begin{align}
    \label{eq:p_i-matroid}
        p_i = \max\limits_{A\in\mathcal{I}}\frac{f_i(A)}{F(A)}. 
    \end{align}
Other than different definition for $p_i$'s and different $\kappa$, \Cref{alg:k-sparsification} is the same as \Cref{alg:sparsification}. 

\begin{theorem}\label{thm:k-sparsifier}
    \Cref{alg:k-sparsification} returns a vector $\bmw$ with expected size at most $O(\frac{B\cdot r\cdot n}{\epsilon^2})$ such that, with probability at least $1-\delta$, for $F'=\sum_{i\in D} \bmw_i f_i$ we have
    \[
        (1- \epsilon)F'(S)\leq F(S)\leq (1+ \epsilon)F'(S) \quad \forall S\subseteq \mc{M}.
    \]
    
\end{theorem}

Theorem~\ref{thm:k-sparsifier} proves the existence part of Theorem~\ref{thm:main-body-K-S}. \Cref{alg:k-sparsification} can be turned into a polynomial-time algorithm if one can approximate $p_i$'s~\eqref{eq:p_i-matroid}.
By modifying the proof of Theorem~\ref{thm:unconstrained-pi-approximation} we prove the following.
\begin{theorem}\label{thm:matroid-pi-approximation}
    Let $f$ and $g$ be two monotone submodular functions and $\mc{M}=(E,\mc{I})$ be a matroid. Then there exists a polynomial-time algorithm that approximates $\max_{S\subseteq E,S\in\mc{I}}\frac{f(S)}{g(S)}$ within $O(\sqrt{n}\log n)$ factor. 
\end{theorem}
By this theorem, when all $f_i$s are monotone we can compute $\hat{p}_i$'s with $p_i\leq \hat{p}_i\leq O(\sqrt{n}\log n)p_i$ in polynomial time which leads to a polynomial-time randomized algorithm that constructs an $\epsilon$-sparsifier of the expected size at most $O(\frac{B\cdot r\cdot n^{1.5} \log n}{\epsilon^2})$. This proves the second part of \Cref{thm:main-body-K-S}.


\section{Applications}
\subsection{Submodular function maximization with cardinality constraint}

\begin{algorithm}[t]
  \caption{}\label{alg:k-maximization}
  \begin{algorithmic}[1]
    \Require{Submodular function $F=\sum\limits_{i\in D} f_i$ with each $f_i:\{0,1\}^E\to \bbR$, constant $k$, and $\epsilon,\delta\in(0,1)$}
    \State{Compute $F'=\sum_{i\in D} \bmw_i f_i$, an $\epsilon$-sparsifier for $F$.}
    \State{$A\gets\emptyset$.}
    \While{$|A|\leq k$}
        \State{$a_i\gets\argmax_{a\in E\setminus A} (F'(A\cup\{a\})-F'(A))$.}
        \State{$A\gets A\cup \{a_i\}$.}
    \EndWhile
    \State{\Return $A$.}
  \end{algorithmic}
\end{algorithm}

Our sparsification algorithm can be used as a preprocessing step and once a sparsifier is constructed it can be reused many times (possibly for maximization/minimization under several different constraints). To elaborate on this, we consider the problem of maximizing a submodular function subject to a cardinality constraint. That is finding $S^*=\argmax_{S\subseteq E,|S|\leq k}F(S)$. Cardinality constraint is a special case of matroid constraint where the independent sets are all subsets of size at most $k$ and the rank of the matroid is $k$. A celebrated result of \citet{NemhauserWF78} states that for non-negative monotone submodular functions a simple greedy algorithm provides a solution with $(1 - 1/\mathrm{e})$ approximation guarantee to the optimal (intractable) solution. For a ground set $E$ of size $n$ and a monotone submodular function $F=\sum_{i\in D}f_i$, this greedy algorithm needs $O(k n N)$ function evaluations to find $S$ of size $k$ such that $F(S)\geq (1 - 1/\mathrm{e})F(S^*)$. We refer to this algorithm as \texttt{GreedyAlg}. In many applications where $N\gg n$, having a sparsifier is beneficial. Applying \texttt{GreedyAlg} on an $\epsilon$-sparsifier of size $O(Bk n/\epsilon^2)$ improves the number of function evaluations to $O(Bk^2 n^2/\epsilon^2)$ and yields $S$ of size $k$ such that $F(S)\geq (1 - 1/\mathrm{e}-\epsilon)F(S^*)$ with high probability (see Algorithm~\ref{alg:k-maximization}). 

We point out that sampling techniques such as \cite{Mitrovic0ZK18,MirzasoleimanBK15} sample elements from the ground set $E$ rather than sampling from functions $f_1,\dots,f_N$. Hence their running time depend on $N$, which could be slow when $N$ is large --- the regime we care about. Besides, our algorithm can be used as a preprocessing step for these algorithms. For instance, the lazier than lazy greedy algorithm \cite{MirzasoleimanBK15} requires $O(nN\log{\frac{1}{\epsilon}})$ function evaluations. However, when $N$ is much larger than $n$ it is absolutely beneficial to use our sparsification algorithm and reduce the number of submodular functions that one should consider.

\subsection{Two well-known examples}
\paragraph{Maximum Coverage Problem.} Let $[N]$ be a universe and $E=\{S_1,\dots,S_n\}$ with each $S_i\subseteq N$ be a family of sets. Given a positive integer $k$, in the \textsc{Max Coverage} problem the objective is to select at most $k$ of sets from $E$ such that the maximum number of elements are covered, i.e., the union of the selected sets has maximal size. One can formulate this problem as follows. For every $i\in [N]$ and $A\subseteq [n]$ define $f_i(A)$ as

\begin{align*}
    f_i(A)=
    \begin{cases}
      1 & \text{if there exists $a\in A$ such that $i\in S_a$} ,\\
      0 & \text{otherwise.}
   \end{cases}
\end{align*}
Note that $f_i$'s are monotone and submodular. Furthermore, define $F:2^n\to \bbR_+$ to be $F(A)=\sum_{i\in [N]}f_i(A)$ which is monotone and submodular as well. Now the \textsc{Max Coverage} problem is equivalent to $\max_{A\subseteq [n], |A|\leq k} F(A)$. For each submodular function $f_i$, the corresponding $p_i$ is
\begin{align*}
    p_i&=\max\limits_{\substack{A\subseteq [n], |A|\leq k}}\frac{f_i(A)}{F(A)} = \max\limits_{\substack{S_a\in E, i\in S_a}}\frac{f_i(\{a\})}{F(\{a\})}\\
    &= \max\limits_{\substack{S_a\in E, i\in S_a}}\frac{1}{F(\{a\})}=\max\limits_{\substack{S_a\in E, i\in S_a}}\frac{1}{|S_a|}.
\end{align*}
We can compute all the $p_i$'s in $O(\sum |S_i|)$ time, which is the input size. Then we can construct a sparsifier in $O(N)$ time. In total, the time required for sparsification is $O(\sum|S_i|+N)$. On the other hand, for this case we have
\begin{align*}
    \sum\limits_{i=1}^N p_i =\sum\limits_{i=1}^N \max\limits_{\substack{S_a\in \mathcal{S}, i\in S_a}}\frac{1}{|S_a|} \leq \sum\limits_{i=1}^n \frac{|S_i|}{|S_i|} = n.
\end{align*}
By Lemma~\ref{lem:UpperBound-p_i}, this upper bound provides that our algorithm constructs an $\epsilon$-sparsifier of size at most $O(k n/\epsilon^2)$. \Cref{alg:k-maximization} improves the running time of the \texttt{GreedyAlg} from $O(k n N)$ to $O(k^2 n^2 /\epsilon^2)$. Furthermore, \Cref{alg:k-maximization} returns a set $A$ of size at most $k$ such that $(1-1/\mathrm{e}-\epsilon)\texttt{OPT}\leq F(A)$. (\texttt{OPT} denotes $F(S^*)$ where $S^*=\argmax_{S\subseteq E,|S|\leq k} F(S)$.)
\paragraph{Facility Location Problem.} Let $I$ be a set of $N$ clients and $E$ be a set of facilities with $|E|=n$. Let $c:I\times E\to \bbR$ be the cost of assigning a given client to a given facility. For each client $i$ and each subset of facilities $A\subseteq E$, define $f_i(A)=\max_{j\in A} c(i,j)$. For any non-empty subset $A\subseteq E$, the value of $A$ is given by
\[
    F(A) = \sum\limits_{i\in I}f_i(A)= \sum\limits_{i\in I}\max\limits_{j\in A} c(i,j).
\]
For completeness, we define $F(\emptyset) = 0$. An instance of the
\textsc{Max Facility Location} problem is specified by a tuple $(I, E, c)$. The objective is to choose a subset $A\subseteq E$ of size at most $k$ maximizing $F(A)$. For each submodular function $f_i$, the corresponding $p_i$ is
\[
    p_i=\max\limits_{\substack{A\subseteq E, |A|\leq k}}\frac{f_i(A)}{F(A)} = \max\limits_{\substack{A\subseteq E, |A|\leq k}}\frac{\max\limits_{j\in A}c(i,j)}{F(A)}=\max\limits_{j\in E}\frac{c(i,j)}{ F(\{j\})}
\]
It is clear that $p_i$'s can be computed in $O(|I|\cdot |E|)$ time, which is the input size. In this case, we have
\begin{align*}
    \sum\limits_{i \in I} p_i &=\sum\limits_{i \in I} \max\limits_{j\in E}\frac{c(i,j)}{ F(\{j\})} \leq \sum\limits_{j=1}^{|E|} \frac{\sum\limits_{i \in I} c(i,j)}{ F(\{j\})} = \sum\limits_{j=1}^{|E|} \frac{F(\{j\})}{F(\{j\})}\\
    &= |E| = n.
\end{align*}
Hence, by Lemma~\ref{lem:UpperBound-p_i}, our algorithm construct a sparsifier of size $O(k n/\epsilon^2)$. \Cref{alg:k-maximization} improves the running time of the \texttt{GreedyAlg} from $O(k n N)$ to in $O(k^2 n^2 /\epsilon^2)$. Furthermore, \Cref{alg:k-maximization} returns a set $A$ of size at most $k$ such that $(1-1/\mathrm{e}-\epsilon)\texttt{OPT}\leq F(A)$.

\begin{remark}
\citet{LindgrenWD16} sparsify an instance of the \textsc{Facility Location} problem by zeroing out entries in the cost matrix --- this is not applicable to the general setting. The runtime of the \texttt{GreedyAlg} applied on their sparsified instance is $O(nN/\epsilon)$. This runtime is huge when $N$ is large --- the regime we care about. Moreover, we can first construct our sparsifier and apply the algorithm of \citeauthor{LindgrenWD16} on it. 
\end{remark}

\subsection{Submodular function minimization}

Besides the applications regarding submodular maximization, our sparsification algorithm can be used as a preprocessing step for submodular minimization as well. In many applications of the submodular minimization problem
such as image segmentation \citep{ShanuAS16-cvpr}, Markov random field inference \citep{FixJPZ13,KohliLT09,VicenteKR09}, hypergraph cuts \citep{DVeldtBK20-kdd}, covering functions \citep{StobbeK10}, the submodular function at hand is a decomposable submodular function.  Many of recent advances on decomposable submodular minimization such as \cite{EneNV17,AxiotisK0SV21} have leveraged a mix of ideas coming from both discrete and continuous optimization. Here we discuss that our sparsifying algorithm approximates the so called \emph{Lov\'{a}sz extension}, a natural extension of a submodular function to the continuous domain $[0,1]^n$. 

\paragraph{Lov\'{a}sz extension.} Let $\bmx\in [0,1]^n$ be the vector $(\bmx_1 ,\bmx_2, \dots,\bmx_n)$. Let $\pi:[n]\to[n]$ be a sorting permutation of
$\bmx_1 ,\bmx_2, \dots,\bmx_n$, which means if $\pi(i) = j$, then $\bmx_j$ is the $i$-th largest element in the vector $\bmx$. Hence, $1\geq \bmx_{\pi(1)}\geq \cdots\geq \bmx_{\pi(n)}\geq 0$. Let $\bmx_{\pi(0)}=1$ and $\bmx_{\pi(n+1)}=0$. Define sets $S_0^\pi=\emptyset$ and $S_i^\pi=\{\pi(1),\dots,\pi(i)\}$. The \emph{Lov\'{a}sz extension} of $f$ is defined as follows
    $
        f^L(\bmx) = \sum_{i=0}^{n} (\bmx_{\pi(i)}-\bmx_{\pi(i+1)})f(S_i^\pi).
    $
    It is well-known that $f^L(\bmx)=\max_{\bmy\in \mc{B}(f)}\langle \bmy, \bmx\rangle$.

 For a decomposable submodular function $F=\sum_{i\in D} f_i$, its Lov\'{a}sz extension is 
    \[
        F^L(\bmx) = \sum\limits_{j=0}^{n} \sum\limits_{i\in D} (\bmx_{\pi(j)}-\bmx_{\pi(j+1)})f_i(S_j^\pi).
    \]
 Recall the definition of $p_i$'s \eqref{eq:p_i}, they can be expressed in an equivalent way in terms of permutations as follow  
\begin{align}
    \label{eq:p_i-Lovasz}
        p_i = \max\limits_{A\subseteq E}\frac{f_i(A)}{F(A)} =\max\limits_{\pi}\max\limits_{j\in[n]} \frac{f_i(S_j^\pi)}{F(S_j^\pi)} .
    \end{align}
Furthermore, note that $F^L(\bmx)$ is a linear combination of $F(S)$, $S\subseteq E$. Given these, we prove Algorithm~\ref{alg:sparsification} outputs a sparsifier that not only approximates the function itself but also approximates its Lov\'{a}sz extension. 
\begin{theorem}
\label{thm:spectral-sparsification}
    \Cref{alg:sparsification} returns a vector $\bmw$ with expected size at most $O(\frac{B\cdot n^2}{\epsilon^2})$ such that, with probability at least $1-\delta$, for $F'=\sum_{i\in D} \bmw_i f_i$ it holds that
    \[
         (1-\epsilon)F'^L(\bmx) \leq F^L(\bmx) \leq (1+\epsilon)F'^L(\bmx) \quad \forall \bmx\in [0,1]^n.
    \]
\end{theorem}

\begin{remark}[Relation to spectral sparsification of graphs]
\label{remark:cut-function}
  The cut function of a graph $G=(V,E)$ can be seen as a decomposable submodular function $F(S) = \sum_{e \in E}f_e$, where $f_e(S) = 1$ if and only if $e \cap S \neq \emptyset$  and $e \cap (V \setminus S) \neq \emptyset$.
  The goal of spectral sparsification of graphs~\cite{SpielmanT11} is to preserve the quadratic form of the Laplacian of $G$, which can be rephrased as $\sum_{e \in E}{f_e^L(\bmx)}^2$.
  In contrast, our sparsification preserves $F^L(\bmx) = \sum_{e \in E}{f_e^L(\bmx)}$.
  Although we can construct a sparsifier that preserves $\sum_{e \in E}{f_e^L(\bmx)}^2$ in the general submodular setting, we adopted the one used here because, in many applications where submodular functions are involved, we are more interested in the value of $\sum_{e \in E}{f_e^L(\bmx)}$ than $\sum_{e \in E}{f_e^L(\bmx)}^2$, and the algorithm for preserving the former is simpler than that for preserving the latter.
\end{remark}


Because our algorithm gives an approximation on the Lov\'{a}sz extension, it can be used as a preprocessing step for algorithms working on Lov\'{a}sz extensions such as the ones in \cite{AxiotisK0SV21,EneNV17}. For instance, it improves the running time of \citet{AxiotisK0SV21} from $\widetilde{O}(T_{\rm maxflow}(n,n+N)\log{\frac{1}{\epsilon}})$ to $\widetilde{O}(T_{\rm maxflow}(n,n+\frac{n^2}{\epsilon^2})\log{\frac{1}{\epsilon}})$ in cases where each submodular function $f_i\in D$ acts on $O(1)$ elements of the ground set which implies $B=\max_{i}|\mc{B}(f_i)|$ is $O(1)$. An example of such cases is hypergraph cut functions with $O(1)$ sized hyperedges.

Next we discuss several examples for which computing $p_i$'s is a computationally efficient task, thus achieving a polynomial-time algorithm to construct sparsifiers. Recall that the cut function of a graph $G=(V,E)$ can be seen as a decomposable submodular function. In this case, computing each $p_e$ for an edge $e=st\in E$ is equivalent to finding the minimum $s$-$t$ cut in the graph, which is a polynomial time task. A more general framework is the \emph{submodular hypergraph minimum $s$-$t$ cut} problem discussed in what follows.

\paragraph{Submodular hypergraphs~\cite{LiM18,Yoshida19}.} Let $\mc{H}$ be a hypergraph with vertex set $V$ and set of hyperedges $E$ where each hyperedge is a subset of vertices $V$. A submodular function $f_e$ is associated to each hyperedge $e\in E$. In the submodular hypergraph minimum $s$-$t$ cut problem the objective is  
\begin{align}
    \mathrm{minimize}_{S\subset V}\sum_{e\in E}f_e(e\cap S)
\end{align}
subject to $s\in S, t\in V\setminus S$. This problem has been studied by \citet{VeldtBenssonKleinberg} and its special cases where submodular functions $f_e$ take particular forms have been studied with applications in semi-supervised learning, clustering, and rank learning (see \citet{LiM18,VeldtBenssonKleinberg} for more details). Examples of such special cases include:
\begin{itemize}
    \item Linear penalty: $f_e(S)=\min\{|S|,|e\setminus S|\}$
    \item Quadratic Penalty: $f_e(S)= |S|\cdot|e\setminus S|$
\end{itemize}
We refer to Table 1 in \citet{VeldtBenssonKleinberg} for more examples. These examples are cardinality-based, that is, the value of the submodular function depends on the cardinality of the input set (see Definition 3.2 of \citet{VeldtBenssonKleinberg}). It is known that if all the submodular functions are cardinality-based, then computing the $s$-$t$ minimum cut in the submodular hypergraph can be reduced to that in an auxiliary (ordinary) graph (Theorem~4.6 of \citet{VeldtBenssonKleinberg}), which allows us to compute $p_e$'s in polynomial time. 

\begin{remark}
Our sparsification algorithm can also be used to construct  \emph{submodular Laplacian} based on the Lov\'{a}sz extension of submodular functions. Submodular Laplacian was introduced by \citet{Yoshida19} and has numerous applications in machine learning, including in learning ranking data, clustering based on network motifs~\cite{LiM17}, network analysis \cite{wsdm/Yoshida16}, and etc.
\end{remark}

\section{Experimental results}

 \begin{figure}
    \subfigure[Uber pickup]{\includegraphics[width=.5\hsize]{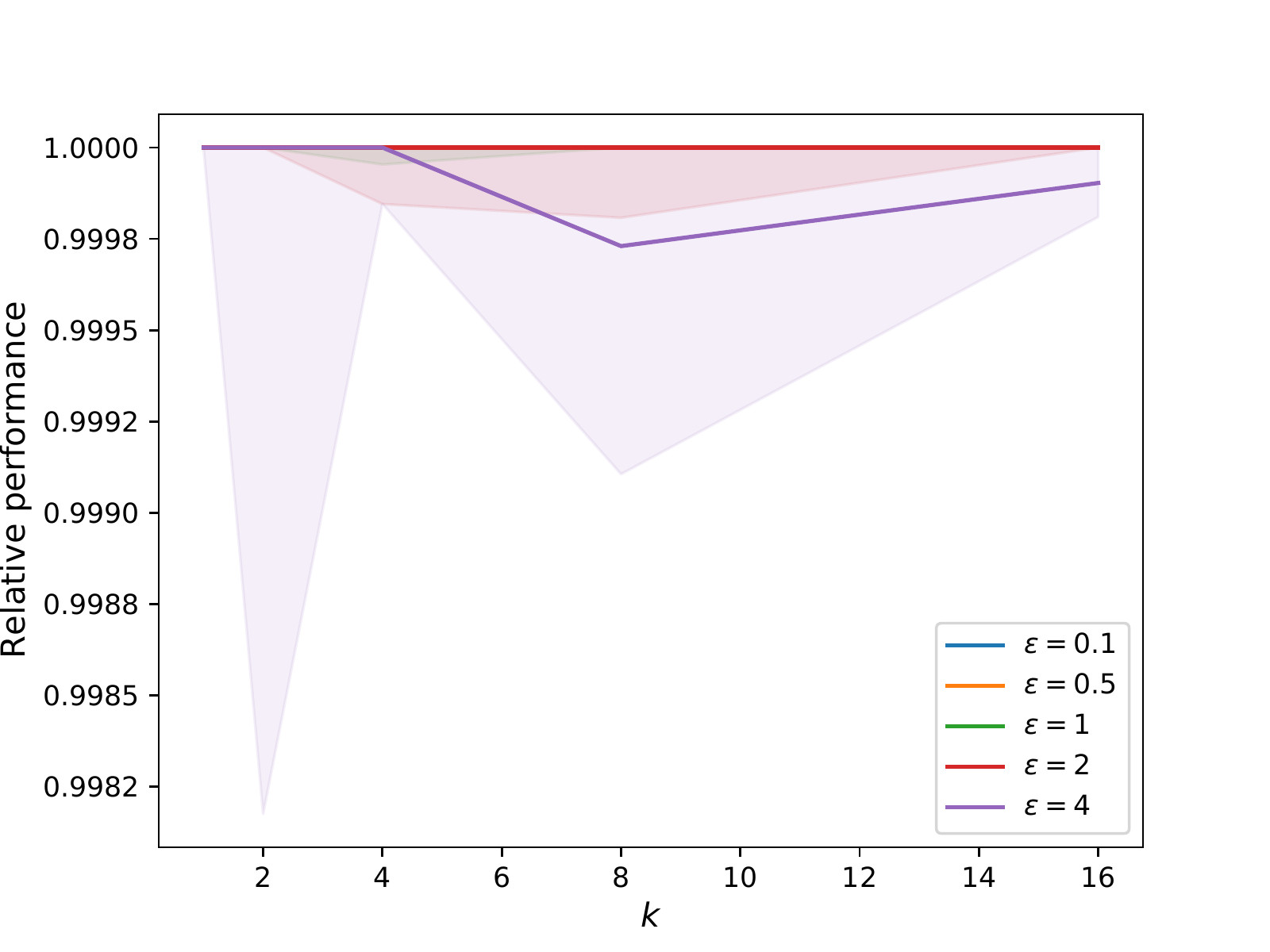}}
    \subfigure[Discogs]{\includegraphics[width=.5\hsize]{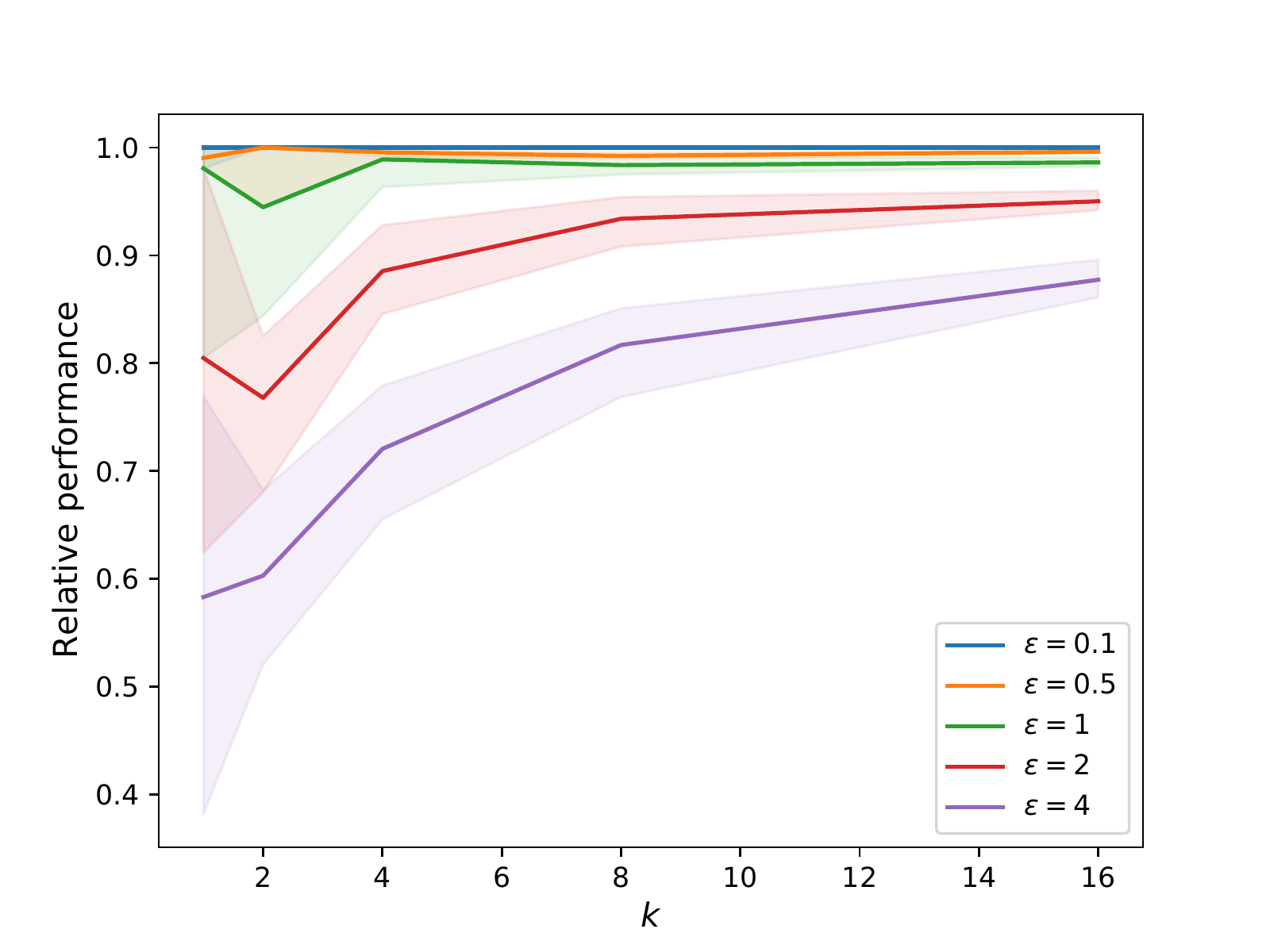}}
    \caption{Relative performance of the greedy method on  sparsifiers.}\label{fig:performance}
   \end{figure}


In this section, we empirically demonstrate that our algorithm (Algorithm~\ref{alg:k-sparsification}) generates a sparse representation of a decomposable submodular function $F:2^E \to \mathbb{R}_+$ with which we can efficiently obtain a high-quality solution for maximizing $F$.
We consider the following two settings.

\textbf{Uber pickup.} We used a database of Uber pickups in New York city in May 2014 consisting of a set $R$ of 564,517 records\footnote{Available at~\url{https://www.kaggle.com/fivethirtyeight/uber-pickups-in-new-york-city}}. Each record has a pickup position, longitude and latitude. Consider selecting $k$ locations as waiting spots for idle Uber drivers.
To formalize this problem, we selected a set $L$ of 36 popular pickup locations in the database, and constructed a facility location function $F:2^L \to \mathbb{R}_+$ as $F(S) = \sum_{v \in R}f_v(S)$, where $f_v(S) = \max_{u \in L} d(u,v) - \min_{u \in S} d(u,v)$ and $d(u,v)$ is the Manhattan distance between $u$ and $v$.
Then, the goal of the problem is to maximize $F(S)$ subject to $|S| \leq k$.

\textbf{Discogs~\cite{KONECT}.}
This dataset provides information about audio records as a bipartite graph $G=(L,R; E)$, where each edge $(u,v) \in L \times R$ indicates that a label $v$ was involved in the production of a release of a style $u$. We have $|L| = 383$ and $|R| = 243,764$, and $|E| = 5,255,950$.
Consider selecting $k$ styles that cover the activity of as many labels as possible.
To formalize this problem, we constructed a maximum coverage function $F:  2^L \to \mathbb{R}$ as $F(S) = \sum_{v \in R}f_v(S)$, where $f_v(S)$ is $1$ if $v$ has a neighbor in $S$ and $0$ otherwise.
Then, the goal is to maximize $F(S)$ subject to $|S| \leq k$.


 \begin{figure}
 \centering
  \includegraphics[width=.5\hsize]{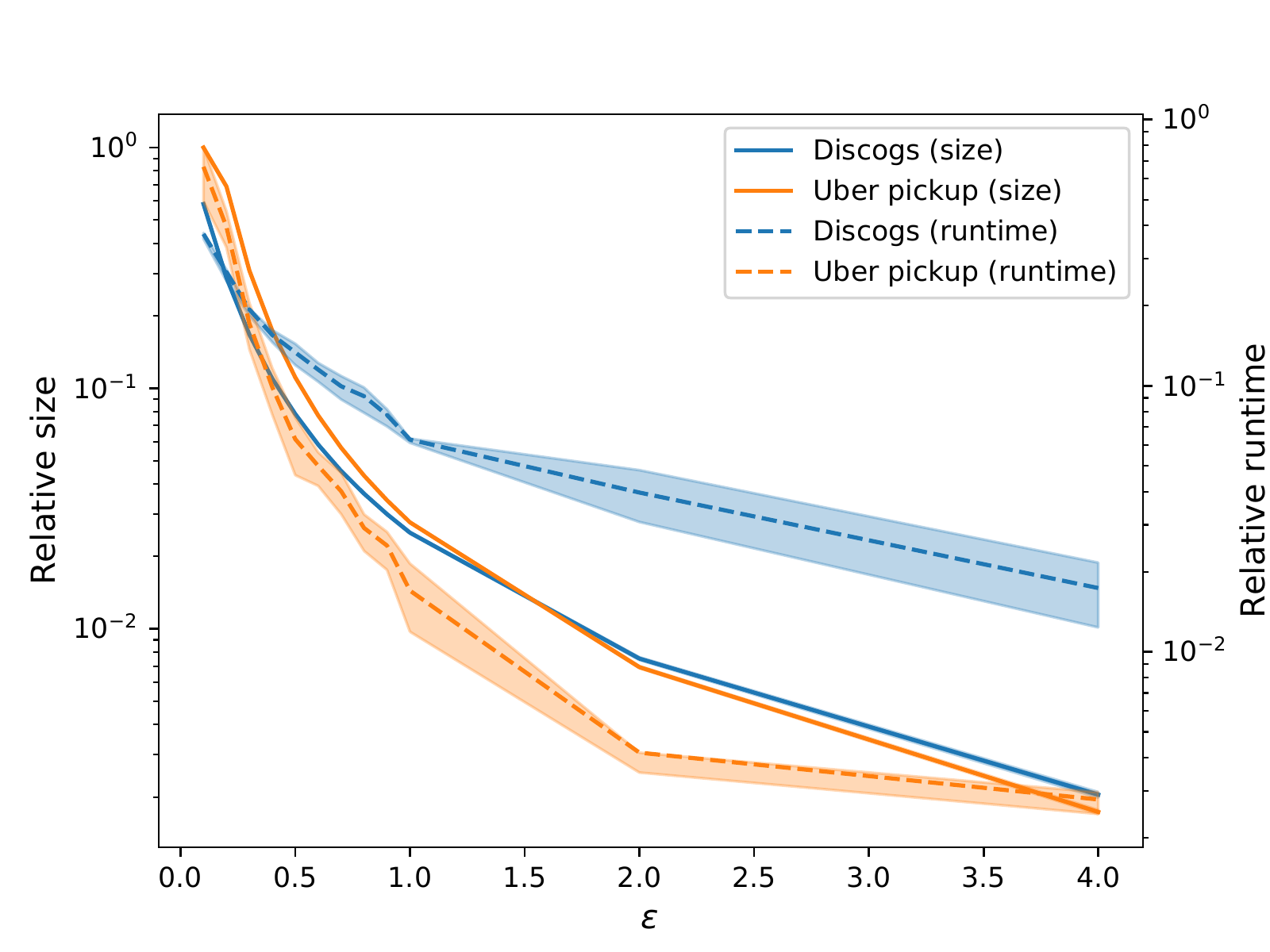}
  \caption{Relative size of sparsifiers and relative runtime of the greedy method on sparsifiers.}\label{fig:size}
   \end{figure}

Figure~\ref{fig:performance} shows the objective value of the solution obtained by the greedy method on the sparsifier relative to that on the original input function with its 25th and 75th percentiles.
Although our theoretical results do not give any guarantee when $\epsilon > 1$, we tried constructing our sparsifier with $\epsilon > 1$ to see its performance.
The solution quality of our sparsifier for Uber pickup is more than 99.9\% even when $\epsilon = 4$, and that for Discogs is more than 90\% performance when $\epsilon \leq 1.0$.
The performance for Uber pickup is higher than that for Discogs because the objective function of the former saturates easily.
These results suggest that we get a reasonably good solution quality by setting $\epsilon = 1$.

\paragraph{Number of functions and speedups.}
Figure~\ref{fig:size} shows the size, that is, the number of functions with positive weights, of our sparsifier relative to that of the original function and the runtime of the greedy method on the sparsifier relative to that on the original function with their 25th and 75th percentiles when $k = 8$.
The size and runtime are decreased by a factor of 30--50 when $\epsilon = 1$.
To summarize, our experimental results suggest that our sparsifier highly compresses the original function without sacrificing the solution quality.

\section*{Conclusion}
Decomposable submodular functions appear in many data intensive tasks in machine learning and data mining. Thus, having a sparsifying algorithm is of great interest. In this paper, we introduce the notion of sparsifier for decomposable submodular functions. We propose the first sparsifying algorithm for these types of functions which outputs accurate  sparsifiers with expected size independent of the size of original function. Our algorithm can be adapted to sparsify mixtures of submodular functions.
We also study the effectiveness of our algorithm under various constraints such as matroid and cardinality constraints. Our experimental results complement our theoretical results on real world data. This work does not present any foreseeable societal consequence.

\section{Missing proofs}

\subsection{Proof of Claim \ref{claim:up-p_i}}
\begin{proof}
\begin{align*}
        \sum\limits_{i\in D} p_i &= \sum\limits_{i\in D} \max\limits_{A\subseteq E}\frac{f_i(A)}{F(A)}
        = \sum\limits_{i\in D}\max\limits_{A\subseteq E}\frac{f_i(A)}{\sum\limits_{j \in D}f_j(A)}\\
        &= \sum\limits_{i\in D}\max\limits_{A\subseteq E}\frac{ \max\limits_{\bmy\in \mc{B}(f_i)}\ip{\bmy,\bmone_A}}{\sum\limits_{j\in D}\max\limits_{\bmy\in \mc{B}(f_j)}\ip{\bmy,\bmone_A}}\\
        &\leq \sum\limits_{i\in D}\max\limits_{A\subseteq E}\frac{\sum\limits_{\bmy\in \mc{B}(f_i)}\ip{\bmy,\bmone_A}}{\sum\limits_{j\in D}\frac{1}{|\mc{B}(f_j)|}\sum\limits_{\bmy\in \mc{B}(f_j)}\ip{\bmy,\bmone_A}}\\
        &\leq \sum\limits_{i\in D}\max\limits_{A\subseteq E}\max\limits_{e\in A}\frac{\sum\limits_{\bmy\in \mc{B}(f_i)}\bmy(e)}{\sum\limits_{j\in D}\frac{1}{|\mc{B}(f_j)|}\sum\limits_{\bmy\in \mc{B}(f_j)}\bmy(e)}
        \\
        &= \sum\limits_{i\in D}\max\limits_{e\in E}\frac{\sum\limits_{\bmy\in \mc{B}(f_i)}\bmy(e)}{\sum\limits_{j\in D}\frac{1}{|\mc{B}(f_j)|}\sum\limits_{\bmy\in \mc{B}(f_j)}\bmy(e)}\\
        &\leq \sum\limits_{i\in D}\sum\limits_{e\in E}\frac{\sum\limits_{\bmy\in \mc{B}(f_i)}\bmy(e)}{\sum\limits_{j\in D}\frac{1}{|\mc{B}(f_j)|}\sum\limits_{\bmy\in \mc{B}(f_j)}\bmy(e)}\\
        &= \sum\limits_{e\in E}\sum\limits_{i\in D}\frac{\sum\limits_{\bmy\in \mc{B}(f_i)}\bmy(e)}{\sum\limits_{j\in D}\frac{1}{|\mc{B}(f_j)|}\sum\limits_{\bmy\in \mc{B}(f_j)}\bmy(e)}\\
        &\leq \sum\limits_{e\in E}\max\limits_{j\in D}|\mc{B}(f_j)|\sum\limits_{i\in D}\frac{\sum\limits_{\bmy\in \mc{B}(f_i)}\bmy(e)}{\sum\limits_{j\in D}\sum\limits_{\bmy\in \mc{B}(f_j)}\bmy(e)}\\
        & \leq \sum\limits_{e\in E}\max\limits_{j\in D}|\mc{B}(f_j)|\frac{\sum\limits_{i\in D}\sum\limits_{\bmy\in \mc{B}(f_i)}\bmy(e)}{\sum\limits_{j\in D}\sum\limits_{\bmy\in \mc{B}(f_j)}\bmy(e)}\\
        & = \sum\limits_{e\in E}\max\limits_{j\in D}|\mc{B}(f_j)|
        = n\cdot(\max_{j\in D}|\mc{B}(f_j)|).
    \end{align*}
\end{proof}

\subsection{Proof of Theorem \ref{thm:k-sparsifier}}

\begin{proof}
    The proof is almost identical to the proof of Lemma~\ref{lem:SS}. We prove that for every $S\in \mc{I}$ with high probability it holds that $(1-\epsilon)F'(S) \leq F^L(S) \leq (1+\epsilon)F'(S)$. 
    Observe that by our choice of $p_i$ and $\bmw_i$ we have $\E[F'(S)]=F(S)$, for all subsets $S\in \mc{M}$. Consider a subset $S_k$. Using \Cref{lem:Chernoff}, we have
    \begin{align}
        &\bbP\left[ |F'(S_k)-\E[F'(S_k)]| \geq \epsilon\E[F'(S_k)]\right]\\ &= \bbP\left[ |F'(S_k)-F(S_k)| \geq \epsilon F(S_k)\right]\\
        &\leq 2\exp{\left(\nicefrac{-\epsilon^2 F(S_k)}{3a}\right)}
        \label{eq:k-concentration}
    \end{align}
    where $a=\max_{i}\bmw_i f_i(S_k)$. We bound the right hand side of~\eqref{eq:k-concentration} by providing an upper bound for $a$.
    \begin{align*}
    \label{eq:k-upper-bound-a}
        a &=\max\limits_{i}\bmw_i f_i(S_k) = \max\limits_{i} \frac{f_i(S_k)}{\kappa\cdot p_i}\\
        &=\max\limits_{i} \frac{f_i(S_k)}{\kappa\cdot\max\limits_{A\in\mc{I}} \frac{f_i(A)}{F(A)}}
        \\
        &\leq \max\limits_{i} \frac{f_i(S_k)}{ \kappa\cdot\frac{f_i(S_k)}{F(S_k)}}
        = \frac{F(S_k)}{\kappa}.
    \end{align*}
    Given the above upper bound for $a$ and the inequality in~\eqref{eq:k-concentration} yields
    \begin{align*}
         &\bbP\left[ |F'(S_k)-F(S_k)| \geq \epsilon F(S_k)\right] \leq 2\exp{\left(-\frac{\epsilon^2 F(S_k)}{3a}\right)} \\
         &\leq 2\exp{\left(-\frac{\epsilon^2 F(S_k)}{3 F(S_k)/\kappa}\right)} = 2\exp{\left(\nicefrac{-\kappa\epsilon^2}{3}\right)}
    \end{align*}
    Recall that $\kappa= 3\log (2n^{r+1}/\delta)/\epsilon^2$. 
    Note that there are at most $n^r$ sets in a matroid of rank $r$. Taking a union bound over all $n^r$ subsets yields that Algorithm~\ref{alg:k-sparsification} with probability at least $1-\delta$ returns a sparsifier for $F$ over the matroid. Similar to Lemma~\ref{lem:UpperBound-p_i}, $\sum_{i\in D}p_i\leq n\cdot(\max_{i\in D}|B(f_i)|)$, and having $\kappa=3\log (2n^{r+1}/\delta)/\epsilon^2$ gives
    \[
        \E[\mathrm{size}(\bmw)] \leq \sum_i \kappa p_i\leq O\left(\frac{rn}{\epsilon^2} \cdot \max\limits_{i\in D}|B(f_i)|\right).
    \]
\end{proof}

\subsection{Proof of Theorem \ref{thm:matroid-pi-approximation}}

\begin{proof}
  The proof is almost the same as that of Theorem~3.5 (Theorem~3.5 in~\cite{BaiIWB16}).
  Therefore, we only explain modifications we need to handle a matroid constraint.

  In the algorithm used in Theorem~3.5,
  given a monotone modular function $f:2^E \to \mathbb{R}_+$, a monotone submodular function $g:2^E \to \mathbb{R}_+$, and a threshold $c \in \mathbb{R}_+$, we iteratively solve the following problem:
  \begin{align}
    \begin{array}{ll}
    \text{minimize} & f(X), \\
    \text{subject to} & g(X) \geq c.
    \end{array}\label{eq:unconstrained}
  \end{align}
  We say that an algorithm for solving~\eqref{eq:unconstrained} a \emph{$(\sigma,\rho)$-bicriterion algorithm} if it outputs $X \subseteq E$ such that $f(X) \leq \sigma f(X^*)$ and $g(X) \geq \rho c$, where $X^*$ is the optimal solution.
  It is shown in~\cite{BaiIWB16} that a $(\sigma,\rho)$-bicriterion algorithm for constant $\sigma$ and $\rho$ leads to an $O(\sqrt{n}\log n)$-approximation algorithm for maximizing $g(X) / f(X)$.

  If we have an additional matroid constraint $\mathcal{M} = (E, \mathcal{I})$, we need a bicriterion algorithm for the following problem:
  \begin{align}
    \begin{array}{ll}
    \text{minimize }& f(X), \\
    \text{subject to} & g(X) \geq c, \\
    & X \in \mathcal{I}.
    \end{array}\label{eq:matroid-constrained}
  \end{align}
  To solve~\eqref{eq:matroid-constrained}, we consider the following problem.
  \begin{align}
    \begin{array}{ll}
    \text{maximize }& g(X), \\
    \text{subject to} & f(X) \leq d, \\
    & X \in \mathcal{I}.
    \end{array}\label{eq:matroid-constrained-dual}
  \end{align}
  This problem is a monotone submodular function maximization problem subject to an intersection of a matroid constraint and a knapsack constraint (recall that $f$ is modular), and is known to admit $\alpha$-approximation for some constant $\alpha$~\cite{Vondrak2011}.
  Then by computing an $\alpha$-approximate solution $X$ for every $d$ of the form $2^i$, and take the minimum $d$ such that $g(X) \geq \alpha \cdot c$, we obtain a $(1,\alpha)$-bicriterion approximation to~\eqref{eq:matroid-constrained}, as desired.
\end{proof}

\subsection{Proof of Theorem~\ref{thm:spectral-sparsification}}
\begin{proof}
It follows from the fact that $F^L(\bmx)$ is a linear combination of $F(S)$, $S\subseteq E$. More precisely, for a decomposable submodular function $F=\sum_{i\in D} f_i$, its Lov\'{a}sz extension is 
    \begin{align}
        F^L(\bmx) &= \sum\limits_{j=0}^{n} \sum\limits_{i\in D} (\bmx_{\pi(j)}-\bmx_{\pi(j+1)})f_i(S_j^\pi)\\
        &=  \sum\limits_{j=0}^{n}  (\bmx_{\pi(j)}-\bmx_{\pi(j+1)})\sum\limits_{i\in D} f_i(S_j^\pi)\\
        &= \sum\limits_{j=0}^{n}  (\bmx_{\pi(j)}-\bmx_{\pi(j+1)}) F(S_j^\pi)
        \label{eq:spectral-equality}
    \end{align}
 Now since our sparsifier approximates $F(S)$ for all subsets $S\subseteq E$ we have
 
 \begin{align}
 \label{eq:approx-S-j}
    (1-\epsilon)F'(S_j^\pi) \leq F(S_j^\pi) \leq (1+\epsilon)F'(S_j^\pi)   \quad 0\leq j\leq n
 \end{align}
 Finally, \eqref{eq:spectral-equality} and \eqref{eq:approx-S-j} yield the following
 \begin{align*}
     (1-\epsilon)F'^L(\bmw) \leq F^L(\bmx) \leq (1+\epsilon)F'^L(\bmx)\quad \forall \bmx\in [0,1]^n.
 \end{align*}
\end{proof}

\appendix
\section*{Acknowledgments}
Akbar Rafiey was supported by NSERC. Yuichi Yoshida was
supported by  JST PRESTO Grant Number JPMJPR192B and JSPS KAKENHI Grant Number 20H05965.
\bibliographystyle{plainnat}
\label{sec:reference_examples}
\bibliography{References}
\clearpage
\newpage

\end{document}